\documentclass[pra]{revtex4}
\usepackage{amsmath}
\usepackage{amsthm}
\usepackage{graphicx}
\usepackage[usenames]{color}

\newcommand{\vnorm}[1]{\left|\left|#1\right|\right|}

\newtheorem{thm}{Theorem}[section]

\newtheorem{cor}[thm]{Corollary}

\begin{document}

\title{Quadratic fermionic interactions yield effective Hamiltonians
  for adiabatic quantum computing}
\date{\today}

\author{Michael J. O'Hara}
\email[Email: ]{mjohara@gmail.com}
\affiliation{Applied Mathematics Program, University of Maryland, 
College Park MD 20742}

\author{Dianne P. O'Leary}
\email[Email: ]{oleary@cs.umd.edu}
\affiliation{Department of Computer Science and Institute for Advanced
Computer Studies, University of Maryland, College Park MD 20742; and
National Institute of Standards and Technology, Mathematical and
Computational Sciences Division, Gaithersburg MD 20899}

\begin{abstract}
    Polynomially-large ground-state energy gaps are rare in many-body
    quantum systems, but useful for adiabatic quantum computing.  We
    show analytically that the gap is generically polynomially-large
    for quadratic fermionic Hamiltonians.  We then prove that
    adiabatic quantum computing can realize the ground states of
    Hamiltonians with certain random interactions, as well as the
    ground states of one, two, and three-dimensional fermionic
    interaction lattices, in polynomial time.  Finally, we use the
    Jordan-Wigner transformation and a related transformation for
    spin-3/2 particles to show that our results can be restated using
    spin operators in a surprisingly simple manner.  A direct
    consequence is that the one-dimensional cluster state can be found
    in polynomial time using adiabatic quantum computing.
\end{abstract}


\maketitle

\section{Introduction}

Adiabatic quantum computing (AQC) \cite{farhi01} is an approach to
quantum computation where a problem is encoded as the ground state of
some Hamiltonian $\mathcal{H}_P$.  It is assumed that it is feasible
to prepare a physical system in the ground state of some simple
Hamiltonian $\mathcal{H}_0$, and then evolve the Hamiltonian slowly
from $\mathcal{H}_0$ to $\mathcal{H}_P$. Under the right conditions
and if the evolution is done sufficiently slowly \cite{ohara08}, then
at the end of the evolution the state of the system is the ground
state of $\mathcal{H}_P$.  Measurement of this final state reveals the
solution to the original problem.  As an approach to quantum
computing, AQC is known to be equivalent to standard gated quantum
computing, in that each can be efficiently simulated by the other
\cite{aharonov07, vandam02}.  Also, a simple Hamiltonian evolution
corresponds to the Grover search algorithm \cite{roland02}.

AQC succeeds in polynomial time if the inverse of the ground-state
energy gap is bounded by a polynomial in the problem size. A typical
Hamiltonian must fit exponentially-many energy levels into a
polynomial-sized energy range, so most energy gaps must be
exponentially small.  It is not clear {\em a priori} why the
ground-state energy gap should ever be larger than the rest.  Let us
be more precise.  Mathematically, we can construct a Hamiltonian with
any given set of $2^n$ energy levels.  Theorem~\ref{thm::rarity}
(below) establishes that large ground-state energy gaps are rare among
choices of energy levels.
\begin{thm}[Large ground-state energy gaps are rare]
  \label{thm::rarity}
  Consider uniform random choices of energy levels for a Hermitian
  operator on $n$ qubits, under the restriction that the ground-state
  energy is zero and the energies are contained in the unit interval.
  The fraction of these choices with a ground-state energy gap greater
  than $2^{-n/2}$ tends to $e^{-2^{n/2}}$ for large $n$.
  \begin{proof}
    The ground-state energy gap is larger than some $\epsilon > 0$
    provided the $2^n-1$ non-zero energy levels are selected from the
    interval $(\epsilon, 1]$.  Thus the fraction of choices of
    energy levels with a ground-state energy gap of at least $\epsilon$ is
    $(1-\epsilon)^{2^n-1}$.  Let us choose $\epsilon=2^{-n/2}$, then
    we have
    \begin{align}
      \left( 1-2^{-n/2} \right)^{2^n-1} &= 
      \left( 1-2^{-n/2} \right)^{\sqrt{2^n-1}\sqrt{2^n-1}} \\
      &\approx \left(e^{-1}\right)^{\sqrt{2^n-1}} \\
      &\approx e^{-2^{n/2}} \;.
    \end{align}
  \end{proof}
\end{thm}

In fact, since the dimension of the problem is exponentially large in
the number of qubits, it is difficult to even {\em determine} the
minimum ground-state energy gap for large problems.  Therefore, it
would be useful to identify a class of Hamiltonians meeting the
following requirements:
\begin{enumerate}
\item The class should have large ground-state energy gaps.
  \label{req::GndGap}
\item The class should have many degrees of freedom and allow many
  interactions between qubits, so that it represents diverse
  problem-instances for AQC.  \label{req::DoF}
\item The class should include a simple $\mathcal{H}_0$ for AQC, and
  be convex, so that
  $\mathcal{H}(s)$, where
  \begin{equation}
    \mathcal{H}(s)=(1-s)\mathcal{H}_0 + s\mathcal{H}_P \;,
    \end{equation}
  stays within the class for $s\in [0,1]$.
  \label{req::closure}
\end{enumerate}

We show that quadratic fermionic interactions can be used to define a
class of Hamiltonians meeting all three requirements.  In
Section~\ref{sec::background}, we develop background on the fermionic
commutation relations, necessary for Section~\ref{sec::EffectiveH},
where we identify a class meeting all three requirements.  We show
that, provided the interaction coefficients are bounded, the fraction
of Hamiltonians with a ground-state energy gap greater than $2x/n$ for
any $x>0$ tends to $e^{-x}$ for large $n$. For random Hamiltonians
whose interaction coefficients are constructed from standard normal
Gaussian distributions, we show the ground-state energy gap is
$\mathcal{O}(1/\sqrt{n})$.  In Section~\ref{sec::largegaps}, we
analytically bound the ground-state energy gap for two and three
dimensional lattices of interacting fermions.  We then show that AQC
can find the ground states of certain random Hamiltonians and these
two and three dimensional fermionic lattices in polynomial time.
Finally, in Section~\ref{sec::representation} we derive alternate
representations of these Hamiltonians using the Jordan-Wigner
transform and a related transform for spin-3/2 particles, and show
that the one-dimensional cluster state may be obtained in polynomial
time using adiabatic quantum computation.

\section{Background on the fermionic commutation relations}
\label{sec::background}

For a detailed exposition on properties of the fermionic commutation
relations (FCRs), see, e.g., \cite{nielsen05, blaizot86,
ohara08thesis}.  Here we only highlight some essential points, mostly
without proof, that is needed to develop the results in following
sections.

The FCRs on a set of linear operators $\{c_j: j=1\dots n\}$ are
\begin{align}
  \{c_j, c_k^{\dagger}\} &= \delta_{j,k}\;, & \{c_j, c_k\} &= 0 \;,
  \label{eqn::FCRsB}
\end{align}
where the bracket notation indicates the anti-commutator
$\{x,y\}=xy+yx$, and $\delta_{jk}$ equals the identity operator if
$j=k$ and zero otherwise.  The superscript dagger
denotes the Hermitian adjoint.  A consequence of the
FCRs is that $\{c_j, c_j^{\dagger}: j=1\dots n\}$ are creation and
annihilation operators that anticommute.

Suppose we have a Hamiltonian of the form
\begin{align}
  \mathcal{H} &= \sum_{j=1}^n C_j c_j^{\dagger}c_j \;,
  \label{eqn::SimpleFermiB}
\end{align}
where the coefficients $C_j$ are positive and real.  All the terms in
$\mathcal{H}$ commute, and the $j^{th}$ term has eigenvalues $0$ and
$C_j$.  Now, take the sum of elements in each possible subset
(including the empty set) of $\{C_j:\;j=1\dots n\}$.  The $2^n$
resulting values are the eigenvalues of $\mathcal{H}$.  In particular,
the ground-state energy of $\mathcal{H}$ is zero, and the ground-state
energy gap is the least non-zero coefficient $C_j$.  To decide whether
an arbitrary value is an eigenvalue of $\mathcal{H}$ for arbitrary
coefficients is NP-complete however, as it is equivalent to the
subset-sum problem (also known as the knapsack problem)
\cite{nielsen05}.

We can write many Hamiltonians in the form of
(\ref{eqn::SimpleFermiB}) using Theorem~\ref{thm::lieb} (below),
originally due to Lieb {\em et al.}~\cite{lieb61}.  Suppose we have a
{\em quadratic fermionic Hamiltonian} $\mathcal{H}$, defined as
\begin{equation}
  \mathcal{H} = \sum_{j,k=1}^n 
  A_{j,k} \left(c_j^{\dagger} c_k - c_j c_k^{\dagger}\right)
  + B_{j,k} \left(c_j^{\dagger} c_k^{\dagger} - c_j c_k \right) \;,
  \label{eqn::quadForm}
\end{equation}
for some set of real coefficients $A_{j,k}$ and $B_{j,k}$.  For
convenience, we gather the coefficients $A_{j,k}$ and $B_{j,k}$ into
real $n\times n$ matrices that we label $A$ and $B$.  If $B=0$, then
$\mathcal{H}$ represents a Hubbard model with no on-site interactions,
for instance electrons in metals \cite{devreese79} or
graphene \cite{pachos07}. If $A$ and $B$ are
tridiagonal, then $\mathcal{H}$ represents a one-dimensional chain of
interacting spin-1/2 particles.  If $A$ and $B$ have three non-zero
super- and sub-diagonals, then $\mathcal{H}$ represents a chain of
interacting spin-3/2 particles \cite{dobrov03}.  Also, we can see
using the FCRs that different choices of $A$ and $B$ may represent the
same Hamiltonian.  In particular, for any given Hamiltonian, $A$ can
be chosen to be symmetric and $B$ anti-symmetric.

Theorem~\ref{thm::lieb} establishes that we can write
(\ref{eqn::quadForm}) in the form of (\ref{eqn::SimpleFermiB}), added
to a multiple of the identity.  Thus we can easily find the first few
eigenvalues of $\mathcal{H}$, and in particular its ground-state
energy gap.
\begin{thm}[Lieb {\em et al.}, 1961]
  \label{thm::lieb}
  Consider a quadratic fermionic Hamiltonian as in
  (\ref{eqn::quadForm}), where $A$ is an $n\times n$ real symmetric
  matrix, $B$ is an $n\times n$ real anti-symmetric matrix, and the
  operators $\{c_k:k=1,\dots,n\}$ satisfy the FCRs.  Then we can find
  $\Lambda^2$ diagonal and $X$ unitary so that
  $X(A-B)(A+B)=\Lambda^2X$, and $Y$ unitary so that
  $Y(A+B)(A-B)=\Lambda^2Y$.  Define the operators $\{\eta_j:
  j=1,\dots,n\}$ by
  \begin{align}
    \left.\begin{array}{ll}
      \eta_j &= \frac{1}{2}\sum_{k=1}^n \left(X_{jk}+Y_{jk}\right) c_k 
             + \left(X_{jk}-Y_{jk}\right) c_k^{\dagger} \\
      \eta_j^{\dagger} &= \frac{1}{2}\sum_{k=1}^n 
       \left(X_{jk}-Y_{jk}\right) c_k 
       + \left(X_{jk}+Y_{jk}\right) c_k^{\dagger}
    \end{array}\;\right\}\;
      j = 1\dots n \;.
  \end{align}
  
  Then $\{\eta_j:\;j=1,\dots,n\}$ satisfy the FCRs, and
  \begin{align}
    \mathcal{H} &= \sum_{j=1}^n 2\Lambda_j \eta_j^{\dagger} \eta_j 
    - \left(\sum_{j=1}^n \Lambda_j\right) I_{2^n}\;,
    \label{eqn::princip}
  \end{align}
  where $\Lambda_j$ denotes the $j^{th}$ entry on the diagonal of the
  matrix $\Lambda$ and $I_{2^n}$ is the identity operator.
  \begin{proof}
    See the appendix.
  \end{proof}
\end{thm}

Theorem~\ref{thm::lieb} was used initially by Lieb {\em et al.} to
find the spectrum of the one-dimensional XY model, and subsequently
has been used, for instance, in the analysis of the one-dimensional
model of free electron transport \cite{blaizot86}.  Quadratic
fermionic Hamiltonians as in (\ref{eqn::quadForm}) have also sparked
recent interest because of their application to quantum complexity
theory.  If one takes a set of gates defined by
$U=\exp(i\mathcal{H}t)$ for some $t$ and a constant quadratic
fermionic Hamiltonian $\mathcal{H}$ in the form of
(\ref{eqn::quadForm}), then one obtains a set of gates that resembles
a universal set, but in fact may be classically simulated
\cite{terhal02}.  Broader sets of gates that can be classically
simulated have been identified \cite{somma06, jozsa08}.  To
classically simulate an evolving Hamiltonian, it has been shown
\cite{vandam02} that Hamiltonian evolutions may be efficiently
approximated by discretizing the evolution into a sequence of short,
constant Hamiltonians.  Theorem~\ref{thm::lieb} has also been used to
find efficient sets of quantum gates for computing properties of
quadratic fermionic Hamiltonians in the form (\ref{eqn::quadForm})
\cite{verstraete08}.  Further, the relationship between a vanishing
energy gap and discontinuity in the ground state has been studied for
these Hamiltonians \cite{zanardi07}.

The Jordan-Wigner transformation applied to the Hamiltonian evolution
\begin{equation}
  \mathcal{H}(s) = (1-s) \sum_{j=1}^n \sigma_j^z
  + s \sum_{j=1}^{n-1}  \sigma_j^x \sigma_{j+1}^x
  \label{eqn::1dising}
\end{equation}
transforms $\mathcal{H}(s)$ into a quadratic fermionic Hamiltonian in
the form of (\ref{eqn::quadForm}), thus providing a means for
determining the spectrum for any $s\in[0,1]$.  In fact, this evolution
exhibits a second-order quantum phase transition, and a ground-state
energy gap that decreases as $\mathcal{O}(1/n)$
\cite{zurek05,schutzhold06}, which is surprisingly large. Yet much
broader classes of quadratic fermionic Hamiltonians in the form of
(\ref{eqn::quadForm}) also have remarkably large ground-state energy
gaps.

Since $(A+B)^{\dagger}=A-B$, and since the singular values of a matrix
$M$ are the square roots of the eigenvalues of $M^{\dagger}M$, we see
that $\Lambda_j$ from (\ref{eqn::princip}) is a singular value of
$A+B$. Further, $C_j$ in (\ref{eqn::SimpleFermiB}) can be defined to
be $2\Lambda_j$.  Thus, if $A+B$ is non-singular, then twice the least
singular value is the ground-state energy gap of $\mathcal{H}$.  If
$A+B$ is singular, then $\mathcal{H}$ has a degenerate ground state,
and the least non-zero singular value is the energy gap between the
ground-state subspace and the higher energy levels of the Hamiltonian.
In any case, since $A+B$ has only $n$ eigenvalues, in contrast to
$\mathcal{H}$ which has $2^n$, we might expect that often the least
singular value of $A+B$ is not exponentially small in $n$.  Then the
ground-state energy gap of $\mathcal{H}$ would not be exponentially
small.  In the next section we state and prove more precise
formulations of this claim.

\section{A class of Hamiltonians meeting the requirements of AQC}
\label{sec::EffectiveH}

The class of quadratic fermionic Hamiltonians represented by
(\ref{eqn::quadForm}) in fact meets all three requirements to be
useful for adiabatic quantum computing.  To establish that the
ground-state energy gaps are large (Requirement~\ref{req::GndGap}), we
provide two theorems.  In Theorem~\ref{thm::uniformLargeGaps} (below),
we take a particular distribution of
coefficient matrices $A$ and $B$ under the restriction
$\vnorm{A+B}_2\leq 1$, and establish that the ground-state energy gap
is $\mathcal{O}(1/n)$.  Here all matrix norms are assumed to be the
two-norm, which is equal to the square root of the largest singular
value.  Then, in Theorem~\ref{thm::rndgap}, we show that the
ground-state energy gap is $\mathcal{O}(1/\sqrt{n})$ for
Gaussian-distributed interaction coefficients.

\begin{thm}[Ground-state energy gaps of quadratic 
fermionic Hamiltonians with bounded coefficients]
  \label{thm::uniformLargeGaps}
  Choose a real diagonal $n\times n$ matrix $\Sigma$ uniformly at
  random with entries in the unit interval, and choose $U$ and $V$
  according to any probability distribution over orthogonal $n\times n$
  matrices. Then $C=U\Sigma V^{\dagger}$ represents a
  distribution over all real matrices with $\vnorm{C}_2 \leq 1$.  Take
  $A$ to be the symmetric part of $C$ and $B$ to be the anti-symmetric
  part of $C$, e.g. $A=(C+C^{\dagger})/2$ and $B=(C-C^{\dagger})/2$,
  and let $\mathcal{H}$ be defined as in (\ref{eqn::quadForm}).  The
  probability that the ground-state energy gap of $\mathcal{H}$ is
  greater than $2x/n$, for any $x>0$, tends to $e^{-x}$ for large $n$.
  \begin{proof}
    If the ground-state energy gap $\mathcal{H}$ is greater than
    $2x/n$, then the singular values of $C$ are contained in the
    interval $(x/n,1)$.  The fraction of choices for $\Sigma$ where
    this is true is
    \begin{equation}
      \left(1-\frac{x}{n}\right)^n 
      = \left[\left(1-\frac{x}{n}\right)^{n/x}\right]^x \;,
    \end{equation}
    which tends to $e^{-x}$ for large $n$.
  \end{proof}
\end{thm}

To determine the ground-state energy gap for Gaussian-distributed
interaction coefficients, we first need the following theorem about
random matrices, due to Edelman \cite[Corollary 3.1]{edelman88}.
\begin{thm}[Edelman, 1988]
  \label{thm::edelman}
  Let $C$ be an $n\times n$ matrix, whose elements have independent
  Gaussian distributions with mean zero and unit variance.  We denote
  such distributions as N(0,1).  Let $\varsigma$ be the least singular
  value of $C$.  Then for large $n$, $n\varsigma^2$ converges in
  distribution to
  \begin{equation}
    \rho(x) = \frac{1+\sqrt{x}}{2\sqrt{x}}e^{-(x/2+\sqrt{x})} \;.
    \label{eqn::dist}
  \end{equation}
\end{thm}
Since $n\varsigma^2$ has a probability distribution that is
asymptotically independent of $n$, it follows that
$\varsigma=\mathcal{O}(1/\sqrt{n})$.  Also, (\ref{eqn::dist}) implies
that $\varsigma \neq 0$ with probability
one.  Similar results for other ensembles of random
matrices are known \cite{mehta04}.

Let us now apply Theorem~\ref{thm::edelman} to quadratic fermionic
Hamiltonians in the form of (\ref{eqn::quadForm}).
\begin{thm}[Ground-state energy gaps of quadratic 
fermionic Hamiltonians with Gaussian coefficients]
  \label{thm::rndgap}
  Let $C$ be an $n\times n$ matrix with independent N(0,1)
  coefficients, let $A$ be the symmetric part of $C$, and let $B$ be
  the anti-symmetric part of $C$, so
  \begin{align}
    A &= \frac{C+C^{\dagger}}{2} \;, &
    B &= \frac{C-C^{\dagger}}{2} \;,
  \end{align}
  and $C=A+B$.  Define
  \begin{equation}
    \mathcal{H}= \sum_{j,k=1}^n 
    A_{j,k}\left(c_j^{\dagger}c_k - c_jc_k^{\dagger}\right)+
    B_{j,k}\left(c_j^{\dagger}c_k^{\dagger} - c_jc_k\right) \;,
    \label{eqn::randH}
  \end{equation}
  and let $\gamma$ be the ground-state energy gap of $\mathcal{H}$.
  Then, for large $n$, $n\gamma^2/4$ converges in distribution to
  $\rho(x)$ defined in (\ref{eqn::dist}).
  \begin{proof}
    By Theorem~\ref{thm::edelman}, if $\gamma/2$ is the least singular
    value of $C$, then $n\gamma^2/4$ converges in distribution to
    $\rho(x)$ for large $n$.  Theorem~\ref{thm::edelman} also implies
    that $C$ is non-singular with probability
    one, so $\gamma$ is the ground-state energy gap.
  \end{proof}
\end{thm}

Since $n\gamma^2/4$ has a probability distribution that is
asymptotically independent of $n$, $\gamma=\mathcal{O}(1/\sqrt{n})$.
Recalling Theorem~\ref{thm::rarity}, we see this is a remarkable
property. Since there must be $2^n$ distinct energy levels in an
energy range of $\mathcal{O}(n^2)$, most of the energy gaps must be
exponentially small.  In fact it can be shown that the Hamiltonians in
Theorem~\ref{thm::rndgap} are with probability
one non-degenerate, so these exponentially small gaps
are also non-zero.  Figure~\ref{fig::AllGapHist} illustrates the
difference between the distribution of the ground-state energy gaps
and the rest of the gaps for 1000 randomly-generated 10-qubit
Hamiltonians, and indeed the ground-state energy gaps are typically
much larger than the other gaps.

\begin{figure}
\vskip 0.1in
\begin{center}
\includegraphics[height=2.5in]{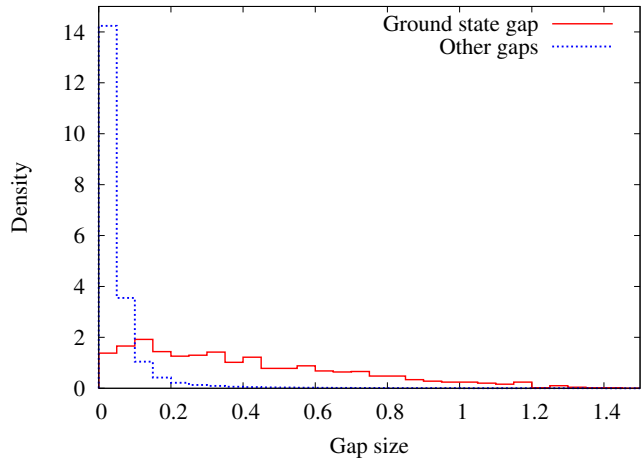}
\end{center}
\caption[Ground-state energy gaps compared to other gaps]{ The
  ground-state energy gap distribution is compared to the distribution
  for the other energy gaps.  All the energy levels are computed for
  1000 random $n=10$ (10-qubit) Hamiltonians.  Each Hamiltonian is
  chosen randomly as described in Theorem~\ref{thm::rndgap}.  As
  predicted by Theorem~\ref{thm::rndgap}, the ground-state energy gaps
  are much larger than the other gaps.}
\label{fig::AllGapHist}
\end{figure}

To establish that quadratic fermionic Hamiltonians as in
(\ref{eqn::quadForm}) have many interaction degrees of freedom
(Requirement~\ref{req::DoF}), we observe that for $n$ qubits there are
$n(n+1)/2$ continuous degrees of freedom in choosing a real symmetric
$n\times n$ matrix $A$, and $n(n-1)/2$ continuous degrees of freedom
in choosing a real anti-symmetric matrix $B$.  The total degrees of
freedom are thus $n^2$.

Finally, it is obvious that if two quadratic fermionic Hamiltonians
are in the form of (\ref{eqn::quadForm}), then their sum is a
quadratic fermionic Hamiltonian in that form
(Requirement~\ref{req::closure}).  We choose, as our initial
Hamiltonian for AQC,
\begin{equation}
  \mathcal{H}_0 = \sum_{j=1}^n \left(2c_j^{\dagger}c_j-I_{2^n}\right) 
   = \sum_{j=1}^n \left(c_j^{\dagger}c_j-c_jc_j^{\dagger}\right) \;.
  \label{eqn::H0}
\end{equation}
The ground state of this Hamiltonian is
easy to construct.  For example, for electrons in a metal, $c_j$ is
the annihilation operator for electron occupation at site $j$, and
then the ground state of $\mathcal{H}_0$ is the state with each site
unoccupied. Then the whole Hamiltonian evolution
$\mathcal{H}(s)$ is a quadratic fermionic Hamiltonian in the form of
(\ref{eqn::quadForm}), where
\begin{equation}
  \mathcal{H}(s)=(1-s)\mathcal{H}_0 + 
  s\sum_{j,k=1}^n 
  \left[A_{j,k}\left(c_j^{\dagger}c_k - c_jc_k^{\dagger}\right)+
  B_{j,k}\left(c_j^{\dagger}c_k^{\dagger} - c_jc_k\right)\right] \;.
  \label{eqn::evolution}
\end{equation}
To find the ground-state energy gap of $\mathcal{H}(s)$, we define
\begin{align}
    \breve{A}(s) &= (1-s)I_{2^n} +sA  \;, \notag \\
    \breve{B}(s) &= s B \;.
    \label{eqn::ABdef}
\end{align}
Then we can rewrite Equation (\ref{eqn::evolution}) as
\begin{equation}
  \mathcal{H}(s)=\sum_{j,k=1}^n 
  \breve{A}_{j,k}(s)\left(c_j^{\dagger}c_k - c_jc_k^{\dagger}\right)+
  \breve{B}_{j,k}(s)\left(c_j^{\dagger}c_k^{\dagger} - c_jc_k\right) \;,
  \label{eqn::ABxform}
\end{equation}
and twice the least non-zero singular value of
$\breve{A}(s)+\breve{B}(s)$ is the ground-state energy gap of
$\mathcal{H}(s)$.

We cannot directly use Theorem~\ref{thm::rndgap} to establish that the
ground-state energy gap is large for {\em all} of $\mathcal{H}(s)$ in
(\ref{eqn::evolution}), since Theorem~\ref{thm::rndgap} is a
probabilistic result for a single random Hamiltonian.  However, for
several special classes of Hamiltonians in (\ref{eqn::evolution}), we
{\em can} establish that the gap is large throughout an evolution.

\section{Hamiltonian evolutions with large ground-state energy gaps}
\label{sec::largegaps}

There are several classes of $A$ and $B$ for which we can easily find
the least singular value of $\breve{A}(s)+\breve{B}(s)$ as defined in
(\ref{eqn::ABdef}), and thus find the minimum ground-state energy gap
of the evolution in (\ref{eqn::evolution}).  We do this for certain
random choices of $A$ and $B$, and also two and three-dimensional
interaction lattices.

If $A+B$ is symmetric positive semi-definite, we can determine the
least singular value of $\breve{A}(s)+\breve{B}(s)$ from the least
singular value of $A+B$.  In that case $B=0$ and thus $\mathcal{H}(1)$
represents a Hubbard model with no on-site interactions.  In order to
choose random samples of symmetric positive semi-definite matrices, we
first choose $C$ to be an $n\times n$ matrix with independent random
N(0,1) elements, as before, and then set $A=CC^{\dagger}$ and $B=0$.

\begin{thm}[Ground-state energy gaps for random Hamiltonian evolutions]
  \label{thm::rndevolve}
  Let $C$ be an $n\times n$ matrix with independent N(0,1)
  coefficients, and let $A=CC^{\dagger}$.  Define
  \begin{equation}
    \mathcal{H}(s)=(1-s)\mathcal{H}_0 + s \sum_{j,k=1}^n 
    A_{j,k}\left(c_j^{\dagger}c_k - c_jc_k^{\dagger}\right)\;.
    \label{eqn::HsymmEvolve}
  \end{equation}
  Then the ground-state energy gap for $\mathcal{H}(s)$ is
  \begin{equation}
    \breve{\gamma}(s) = 2(1-s)+s\gamma \;,
  \end{equation}
  where $n\gamma/2$ converges in distribution for large $n$ to
  $\rho(x)$ defined in (\ref{eqn::dist}).
  \begin{proof}
    Let us label the least singular value of $C$ as $\sqrt{\gamma/2}$.
    Then by Theorem~\ref{thm::edelman}, for large $n$, $n\gamma/2$
    converges in distribution to the density function $\rho(x)$, and
    is non-zero with probability one.
    Since $A=CC^{\dagger}$ is symmetric positive semi-definite,
    $\gamma/2$ is its least singular value.  Thus $\gamma$ is the
    ground-state energy gap of $\mathcal{H}(1)$.

    Define
    \begin{align}
      \breve{A}(s) &= (1-s)I_{n} +sA \;,
    \end{align}
    then the ground-state energy gap $\breve{\gamma}(s)$ of
    $\mathcal{H}(s)$ is twice the least non-zero singular value of
    $\breve{A}(s)$.  Notice $A$ and $I$ are symmetric positive
    semi-definite, and diagonal in the same basis.  Since $\gamma/2$
    is the least eigenvalue of $A$, we have 
    \begin{equation}
      \frac{\breve{\gamma}(s)}{2} = (1-s)
      +s\left(\frac{\gamma}{2}\right) \;.
    \end{equation}
  \end{proof}
\end{thm}

\begin{figure}[!t]
\vskip 0.1in
\begin{center}
\includegraphics[height=3in]{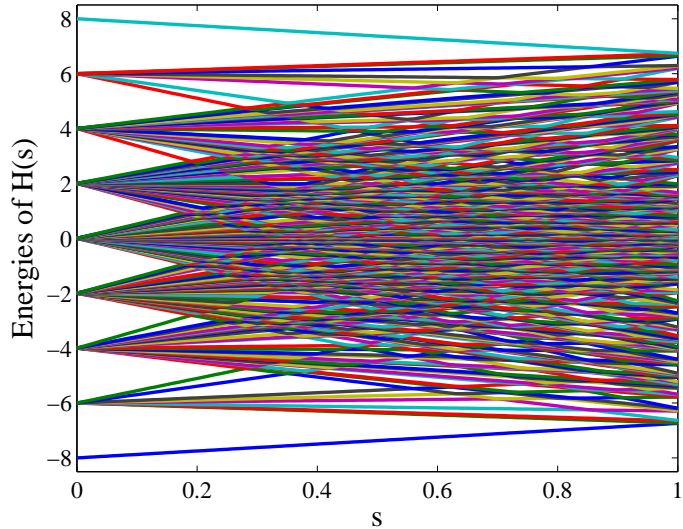}
\end{center}
\caption[Energies for a random Hamiltonian evolution]{Eigenvalues of
  $\mathcal{H}(s)$ as a function of $s$, where $\mathcal{H}(s)$ was
  defined as in Theorem~\ref{thm::rndevolve} with $A=CC^{\dagger}/n$
  for $n=8$.  The division by $n$ is so that $\vnorm{\mathcal{H}(0)}
  \approx \vnorm{\mathcal{H}(1)}$, resulting in a better visualization
  although making the ground-state energy gap $\mathcal{O}(1/n^2)$
  instead of $\mathcal{O}(1/n)$.  We can see that the ground-state
  energy gap is linearly decreasing with $s$ as predicted by
  Theorem~\ref{thm::rndevolve}, and is much larger than most of the
  other energy gaps.}
\label{fig::evolve}
\end{figure}

Figure~\ref{fig::evolve} illustrates the energy levels for an 8-qubit
instance of a random evolution as in (\ref{eqn::HsymmEvolve}).
Evidently the ground-state energy gap is larger than the other gaps
throughout the evolution.

If $A+B$ is not symmetric, then the eigenvalues of $G(s) =
(\breve{A}(s)+\breve{B}(s)) (\breve{A}(s)-\breve{B}(s))$ are of
interest.  We explore three classes of matrices where finding these
eigenvalues is easy, namely, generalizations of the one-dimensional XY
model, and generalizations of two and three-dimensional lattices of
interacting fermions.

The one-dimensional XY model is expressed with $A$ and $B$ matrices
such as \cite[p. 413]{lieb61}:
\begin{align}
  A &= \frac{1}{2}\left( \begin{array}{ccccc}
   0 & 1 & \hphantom{-1} & \hphantom{-1} & 1 \\
   1 & 0 & 1 & & \hphantom{-1}  \\
       & 1 & 0 & \cdot & \\
   \hphantom{-1}  & \hphantom{-1} & \cdot & \cdot & 1 \\
     1 & & & 1 & 0
  \end{array}\right) \;, & 
  B &= \frac{1}{2}\left( \begin{array}{ccccc}
   0 & 1 & \hphantom{-1} & \hphantom{-1} & -1 \\
   -1 & 0 & 1 & & \hphantom{-1}  \\
       & -1 & 0 & \cdot & \\
   \hphantom{-1}  & \hphantom{-1} & \cdot & \cdot & 1 \\
     1 & & & -1 & 0
  \end{array}\right) \;,
  \label{eqn::exCirc}
\end{align}
where omitted entries are zero.  The ground state of the
one-dimensional XY model can be found in polynomial time with AQC
\cite{zurek05,schutzhold06}.  In fact, this holds for more general
choices of $A$ and $B$.  An essential property of the definitions in
(\ref{eqn::exCirc}) is that each row is a cyclic shift of the previous
row.  Such matrices are called {\em circulant}.  For $n$ qubits, it is
easy to check that there are $n$ degrees of freedom in choosing a
symmetric circulant matrix $A$ and anti-symmetric circulant matrix
$B$.  We show that AQC can find the ground state of any Hamiltonian
with circulant $A$ and $B$ matrices in polynomial time.

\begin{thm}[Ground-state energy gaps for circulant $A$ and $B$ matrices]
  \label{thm::circulant}
  Let
  \begin{equation}
    \mathcal{H}= \sum_{j,k=1}^n 
    \left(A_{j,k}\left(c_j^{\dagger}c_k - c_jc_k^{\dagger}\right)+
    B_{j,k}\left(c_j^{\dagger}c_k^{\dagger} - c_jc_k\right)\right) \;,
    \label{eqn::circulantH}
  \end{equation}
  where $A$ is a real circulant $n\times n$ symmetric matrix, and $B$
  is a real circulant anti-symmetric $n\times n$ matrix.  Then the
  ground-state energy gap of $\mathcal{H}$ is bounded from below by a
  polynomial in $1/n$.
  \begin{proof}
    Circulant matrices form a commutative ring \cite[p. 201]{golub96},
    so if $A$ and $B$ are circulant, then so is $G=(A+B)(A-B)$.  Also,
    $G$ is symmetric positive semi-definite, and the ground-state
    energy gap of $\mathcal{H}$ is twice the square root of the least
    non-zero eigenvalue of $G$.

    Circulant matrices also have the nice property that their
    eigenvalues are given by the discrete Fourier transform of their
    first column \cite[p. 124]{chu00}.
    The $n\times n$ Fourier transform matrix is
    \begin{equation}
      F_n =       \frac{1}{\sqrt{n}}
      \left( \begin{array}{cccccc}
	e^{(0\cdot 0)2\pi i/n} & e^{(0\cdot 1)2\pi i/n}  &
	e^{(0\cdot 2)2\pi i/n} & ...\\
	e^{(1\cdot 0)2\pi i/n} & e^{(1\cdot 1)2\pi i/n}  & 
	e^{(1\cdot 2)2\pi i/n} & ...\\
	e^{(2\cdot 0)2\pi i/n} & e^{(2\cdot 1)2\pi i/n}  & 
	e^{(2\cdot 2)2\pi i/n} & ...\\
	... &&&\\
	e^{((n-1)1\cdot 0)2\pi i/n} & e^{((n-1)\cdot 1)2\pi i/n}  & 
	e^{((n-1)\cdot 2)2\pi i/n} & ...
      \end{array}\right) \;.
      \label{eqn::FFT}
    \end{equation}
    Recall we label the eigenvalues of $G$ as $\Lambda_k^2$.  Labeling
    the entries in the first column of $G$ as $g_k$, we can write:
    \begin{align}
      \left( \begin{array}{c}
	\Lambda_1^2 \\
	\Lambda_2^2 \\
	\Lambda_3^2 \\
	...\\
	\Lambda_{n}^2
      \end{array}\right)
      &= F_n
      \left( \begin{array}{c}
	g_1 \\
	g_2 \\
	g_3 \\
	...\\
	g_{n}
      \end{array}\right) \;.
    \end{align}
    Also, $G$ is symmetric so $g_k = g_{n+2-k}$ for $k\geq 2$.  Then
    we have, for $n$ odd:
    \begin{align}
      \Lambda_k^2 &= \frac{1}{\sqrt{n}} \left[ g_1 + \sum_{j=2}^{(n+1)/2}
      g_j\left( e^{((j-1)\cdot (k-1))2\pi i/n} + e^{((n-j+1)\cdot (k-1))2\pi
	i/n}\right)\right] \\
      &= \frac{1}{\sqrt{n}} \left[ g_1 +
      \sum_{j=2}^{(n+1)/2} g_j\left( e^{((j-1)\cdot (k-1))2\pi i/n} 
      + e^{-((j-1)\cdot
	(k-1))2\pi i/n}\right)\right] \\
      &= \frac{1}{\sqrt{n}} \left[ g_1 +
      \sum_{j=2}^{(n+1)/2} 
         g_j\cos\left(\frac{2\pi (j-1)(k-1)}{n}\right)\right] \;.
    \end{align}
    For $n$ even we have a leftover term, but it simplifies:
    \begin{align}
      \Lambda_k^2 &= 
      \frac{1}{\sqrt{n}} \left[
      g_1 + g_{n/2+1}e^{((n/2)\cdot (k-1))2\pi i/n}
      + \sum_{j=2}^{n/2} 
        g_j\cos\left(\frac{2\pi (j-1)(k-1)}{n}\right)  \right] \\
      &= \frac{1}{\sqrt{n}} \left[
      g_1 + (-1)^{(k-1)} g_{n/2+1} +
      \sum_{j=2}^{n/2} 
         g_j\cos\left(\frac{2\pi (j-1)(k-1)}{n}\right)
      \right] \;.
    \end{align}
    For $n\geq 1$, we have $1/n \leq 1/\sqrt{n}$.  So whether $n$ is
    even or odd, Taylor expansion of the
    cosine makes it clear that if $\Lambda_k\neq 0$ then $\Lambda_k$
    is bounded from below by a
    polynomial in $1/n$.  So the ground-state energy gap of
    $\mathcal{H}$ is bounded from below by a polynomial in $1/n$.
  \end{proof}
\end{thm}

We can extend this result to a whole Hamiltonian evolution.
\begin{cor}[Hamiltonian evolutions with circulant $A$ and $B$ matrices]
  \label{cor::circulant}
  Let
  \begin{equation}
    \mathcal{H}(s)=(1-s)\mathcal{H}_0 + 
    s\sum_{j,k=1}^n 
    \left(A_{j,k}\left(c_j^{\dagger}c_k - c_jc_k^{\dagger}\right)+
    B_{j,k}\left(c_j^{\dagger}c_k^{\dagger} - c_jc_k\right)\right) \;,
    \label{eqn::circulantEv}
  \end{equation}
  where $A$ is a real circulant $n\times n$ symmetric matrix, and $B$
  is a real circulant anti-symmetric $n\times n$ matrix.  Then the
  ground-state energy gap of $\mathcal{H}(s)$ is bounded from below by a
  polynomial in $1/n$ and $s$.
  \begin{proof}
      First rewrite $\mathcal{H}(s)$ as in (\ref{eqn::ABxform}).  The
      elements of the matrix 
      \begin{equation}
	\breve{G}(s)=\left(\breve{A}(s)+\breve{B}(s)\right)
	\left(\breve{A}(s)-\breve{B}(s)\right)
      \end{equation}
      are quadratic functions of $s$, and $\breve{A}(s)$ and
      $\breve{B}(s)$ are circulant.  Then we can follow the proof of
      Theorem~\ref{thm::circulant} to see that the eigenvalues of
      $\breve{G}(s)$ are bounded by a polynomial in $s$ and $1/n$,
      and so the ground-state energy gap of $\mathcal{H}(s)$ is
      bounded from below by a polynomial in $1/n$ and $s$.
  \end{proof}
\end{cor}

Circulant coefficient matrices other than the XY model include
scenarios such as non-nearest neighbor interactions on a
one-dimensional chain of interacting fermions.  The restriction that
$A$ and $B$ are circulant imposes the requirement that the interaction
strengths are independent of position, and it imposes periodic
boundary conditions.  It should be noted that while circulant matrices
yield to elegant analysis, these results could be extended to
interaction matrices derived from other boundary conditions.  For
instance, if $G$ is symmetric Toeplitz tridiagonal
(Toeplitz matrices are those with constant diagonals), then its
eigenvalues may be found analytically
\cite[p. 158]{elaydi96}. Analysis of the two and
three-dimensional interaction grids then builds on the one-dimensional
analysis in exact analogy to the circulant case.

Now let us consider the case of a two-dimensional lattice of
interacting fermions with periodic boundary conditions.  The $A$ and
$B$ matrices then have a block structure such as 
\begin{align}
  A &= \left( \begin{array}{ccccc}
      A_0 & I     &       &       & I \\
      I   & A_0   & I     &       & \\
          & \cdot & \cdot & \cdot & \\
          &       &  I    & A_0   & I \\
      I   &       &       & I     & A_0 \\
    \end{array}\right) \;, &
  B &= \left( \begin{array}{ccccc}
      B_0 & I     &       &       & -I \\
      -I   & B_0   & I     &       & \\
          & \cdot & \cdot & \cdot & \\
          &       &  -I    & B_0   & I \\
      I   &       &       & -I     & B_0 \\
    \end{array}\right) \;, 
  \label{eqn::exBCCB}
\end{align}
where $B_0$ and $A_0$ are as in (\ref{eqn::exCirc}).  Evidently each
block is circulant, and $A$ and $B$ are circulant in the blocks.  Such
matrices are called ``block circulant with circulant blocks'' (BCCB).
Let us assume that the blocks are $p\times p$, and there are $q$
blocks per row, so $n=pq$.

\begin{thm}[Ground-state energy gaps for BCCB $A$ and $B$ matrices]
  \label{thm::BCCB}
  Let $A$ and $B$ be $n\times n$ BCCB matrices with $p\times p$ blocks
  and $q$ blocks per row, so $pq=n$.  Assume $A$ is real symmetric and
  $B$ is real anti-symmetric.  Define the quadratic fermionic
  Hamiltonian $\mathcal{H}$ with $A$ and $B$ as in
  (\ref{eqn::quadForm}).  Then the ground-state energy gap of
  $\mathcal{H}$ is bounded from below by a polynomial in $1/n$.

  \begin{proof}
    It is easy to check $G=(A+B)(A-B)$ is BCCB, and since
    $G$ is symmetric, it is
    block symmetric and its blocks are
    symmetric.  Then $T^{\dagger}GT$ is diagonal
    \cite{swarztrauber77}, where
    \begin{equation}
      T = \left( \begin{array}{ccc}
	  F_p &     &  \\
	      & F_p &  \\
  	      &     & \ddots       
	\end{array}\right) P_y 
      \left( \begin{array}{ccc}
	  F_q &     &  \\
	      & F_q &  \\
  	      &     & \ddots       
	\end{array}\right) \;,
      \label{eqn::defT}
    \end{equation}
    and $P_y$ is a permutation matrix that re-orders the columns as
    $(1,\; p+1,\; 2p+1,\; \dots,\; (q-1)p+1,\; 2,\; p+2,\; 2p+2,\;
    \dots)$.  The matrix $T$ diagonalizes $G$ by first diagonalizing
    the blocks of $G$, then reordering the rows and columns so that
    the matrix is block diagonal with circulant blocks, and finally
    diagonalizing those blocks.  Geometrically, this procedure can be
    thought of as a Fourier transform first along the horizontal axis
    of the lattice and then along the vertical axis.
    
    The eigenvalues of $G$ are then the eigenvalues of symmetric
    circulant matrices, whose entries are the eigenvalues of symmetric
    circulant matrices.  Following the proof of
    Theorem~\ref{thm::circulant}, we can see that the non-zero
    eigenvalues of $G$ are bounded from below by a polynomial in $1/n$.
  \end{proof}
\end{thm}

We can extend this result to a whole Hamiltonian evolution:
\begin{cor}[Hamiltonian evolutions with BCCB $A$ and $B$ matrices]
  \label{cor::BCCB}
  Let
  \begin{equation}
    \mathcal{H}(s)=(1-s)\mathcal{H}_0 + 
    s\sum_{j,k=1}^n 
    \left[A_{j,k}\left(c_j^{\dagger}c_k - c_jc_k^{\dagger}\right)+
    B_{j,k}\left(c_j^{\dagger}c_k^{\dagger} - c_jc_k\right)\right] \;,
  \end{equation}
  where $A$ is a real BCCB $n\times n$ symmetric matrix, and $B$
  is a real BCCB anti-symmetric $n\times n$ matrix.  Then the
  ground-state energy gap of $\mathcal{H}(s)$ is bounded from below by a
  polynomial in $1/n$ and $s$.
  \begin{proof}
    The proof is analogous to the proof of
    Corollary~\ref{cor::circulant}.
  \end{proof}
\end{cor}

We can even do the three-dimensional lattice of interacting fermions.
Then the $A$ and $B$ matrices are as in (\ref{eqn::exBCCB}), but $A_0$
and $B_0$ are BCCB instead of circulant.  Let us call these matrices
``block circulant with BCCB blocks'', or
(BC)$^2$CB.

\begin{thm}[Ground-state energy gaps for (BC)$^2$CB $A$ and $B$ matrices]
  \label{thm::BC2CB}
  Let $A$ and $B$ be $n\times n$ (BC)$^2$CB matrices.  Let the number
  of blocks be $r$, and the BCCB blocks contain $q$ circulant
  subblocks each $p\times p$, so $n=pqr$.  Define the quadratic
  fermionic Hamiltonian $\mathcal{H}$ with $A$ and $B$ as in
  (\ref{eqn::quadForm}).  Then the ground-state energy gap of
  $\mathcal{H}$ is bounded from below by a polynomial in $1/n$.

  \begin{proof}
    The proof is analogous to that of Theorem~\ref{thm::BCCB}, but we
    set
    \begin{equation}
      T = \left( \begin{array}{ccc}
	  F_p &     &  \\
	      & F_p &  \\
  	      &     & \ddots       
	\end{array}\right) P_y 
      \left( \begin{array}{ccc}
	  F_q &     &  \\
	      & F_q &  \\
  	      &     & \ddots       
	\end{array}\right) P_z 
      \left( \begin{array}{ccc}
	  F_r &     &  \\
	      & F_r &  \\
  	      &     & \ddots       
	\end{array}\right) \;,
      \label{eqn::defT3}
    \end{equation}
    where $P_z$ is a permutation matrix that re-orders the columns as
    $(1,\; pq+1,\; 2pq+1,\; \dots,\; (r-1)pq+1,\; 2,\; pq+2,\;
    2pq+2,\; \dots)$.  This transformation first diagonalizes the BCCB
    blocks, then permutes rows and columns to obtain a block-diagonal
    matrix with circulant blocks, and diagonalizes the remaining
    blocks.
  \end{proof}
\end{thm}

We can extend this result to a whole Hamiltonian evolution:
\begin{cor}[Hamiltonian evolutions with (BC)$^2$CB $A$ and $B$ matrices]
  \label{cor::BCBCCB}
  Let
  \begin{equation}
    \mathcal{H}(s)=(1-s)\mathcal{H}_0 + 
    s\sum_{j,k=1}^n 
    \left[A_{j,k}\left(c_j^{\dagger}c_k - c_jc_k^{\dagger}\right)+
    B_{j,k}\left(c_j^{\dagger}c_k^{\dagger} - c_jc_k\right)\right] \;,
  \end{equation}
  where $A$ is a real (BC)$^2$CB $n\times n$ symmetric matrix, and $B$
  is a real (BC)$^2$CB anti-symmetric $n\times n$ matrix.  Then the
  ground-state energy gap of $\mathcal{H}(s)$ is bounded from below by a
  polynomial in $1/n$ and $s$.
  \begin{proof}
    The proof is analogous to that of
    Corollary~\ref{cor::circulant}.
  \end{proof}
\end{cor}

\section{Representations of the Hamiltonians}
\label{sec::representation}

Using the Jordan-Wigner transformation, we can define Hamiltonians
using other kinds of particle operators that also have large
ground-state energy gaps.  In the Hubbard model of free electrons,
interaction terms such as $c_j^{\dagger}c_k^{\dagger} - c_jc_k$ do not
occur because they do not conserve the number of electrons.  However,
they may occur in spin systems
transformed into fermionic representations.  The best-known example is
the Hamiltonian resulting from the Jordan-Wigner transformation
applied to the XY-model \cite{lieb61}.  Let us first identify {\em
all} the Hamiltonians that, under the Jordan-Wigner
transformation \cite{somma02}
\begin{align}
  c_j &= (-1)^{j-1} \sigma_1^z\sigma_2^z...\sigma_{j-1}^z 
  \left(\frac{\sigma_j^x - i\sigma_j^y}{2}\right) \;, \notag \\
  c_j^{\dagger} &= (-1)^{j-1} \sigma_1^z\sigma_2^z...\sigma_{j-1}^z
  \left(\frac{\sigma_j^x + i\sigma_j^y}{2}\right)\;,
  \label{eqn::JW}
\end{align}
yield a quadratic fermionic Hamiltonian in the form of
(\ref{eqn::quadForm}).  Theorem~\ref{thm::equiv} is equivalent to the
result in \cite[p. 4]{wehefritz06}, but using a different basis
representation.

\begin{thm}[Quadratic fermionic Hamiltonians represented with Pauli operators]
  \label{thm::equiv}
  There is a bijection between Hamiltonians on $n$ qubits of the form
  \begin{equation}
    \mathcal{H} =
    \sum_{j=1}^n W_{j,j} \sigma_j^z 
    + \sum_{k>j} W_{j,k} \sigma_j^x \sigma_{j+1}^z \dots
                  \sigma_{k-1}^z \sigma_k^x 
    + \sum_{k>j} W_{k,j} \sigma_j^y \sigma_{j+1}^z \dots
                  \sigma_{k-1}^z \sigma_k^y  \; ,
    \label{eqn::allterms}
  \end{equation}
  where the coefficients $W_{j,k}$ are real, and Hamiltonians of the
  form
  \begin{equation}
    \mathcal{H} = \sum_{j,k=1}^n 
    A_{j,k} \left(c_j^{\dagger} c_k - c_j c_k^{\dagger}\right) 
    + B_{j,k}\left(c_j^{\dagger} c_k^{\dagger} - c_j c_k \right) \;,
    \label{eqn::pfquad}
  \end{equation}
  where $\{c_j:\;j=1,\dots,n\}$, defined by (\ref{eqn::JW}), satisfy
  the FCRs, $A$ is a real symmetric $n\times n$ matrix, and $B$ is a
  real anti-symmetric matrix.  The bijection is given by the
  invertible transformation
    \begin{align}
    A_{j,j} &= W_{j,j} \;, \notag \\
    A_{j,j+m} &=  A_{j,j+m} 
       = \frac{(-1)^{m+1}}{2}\left(W_{j,j+m}+W_{j+m,j}\right) \;, \notag \\
    B_{j,j+m} &= -B_{j,j+m} 
       = \frac{(-1)^{m+1}}{2}\left(W_{j,j+m}-W_{j+m,j}\right) \;.
  \end{align}
  \begin{proof}
    Apply (\ref{eqn::JW}) to (\ref{eqn::pfquad}), and use the
    commutation relations for Pauli operators to simplify the result.
  \end{proof}
\end{thm}

Using Theorem~\ref{thm::equiv}, we can restate earlier results in a
surprisingly simple manner.  First, observe that application of
Theorem~\ref{thm::equiv} to $\mathcal{H}_0$ defined in (\ref{eqn::H0})
yields
\begin{equation}
  \mathcal{H}_0 = \sum_{j=1}^n \sigma_j^z \;.
\end{equation}
The ground state of $\mathcal{H}_0$ is the
configuration with each particle in a spin-down eigenstate of
$\sigma^z$. If $\mathcal{H}_P$ is in the form of
(\ref{eqn::allterms}), then so is the Hamiltonian evolution
\begin{equation}
  \mathcal{H}(s) = (1-s)\sum_{j=1}^n \sigma_j^z + s\mathcal{H}_P \;,
  \label{eqn::pauliEv}
\end{equation}
for $0\leq s\leq 1$.

Next, observe that, up to sign, the elements of the matrix $W$ are the
same those of $A+B$.  So to find the ground-state energy gap for a
Hamiltonian that can be written in the form of (\ref{eqn::allterms}),
we only need to apply the necessary sign changes to the elements of
$W$, find the least non-zero singular value of the resulting matrix,
and multiply by two.  Thus Theorem~\ref{thm::rndgap} can be applied to
the Hamiltonians in (\ref{eqn::allterms}) yielding a simple result:
\begin{thm}[Ground-state energy gaps of Hamiltonians 
defined using Pauli operators with Gaussian coefficients]
  \label{thm::rndgap2}
  Let $\mathcal{H}$ be defined by (\ref{eqn::allterms}), where the
  elements of $W$ are N(0,1) and independent.  Let $\gamma$ be the
  ground-state energy gap of $\mathcal{H}$.  Then, for large $n$,
  $n\gamma^2/4$ converges in distribution to the probability density
  function
  \begin{equation}
    \rho(x) = \frac{1+\sqrt{x}}{2\sqrt{x}}e^{-(x/2+\sqrt{x})} \;.
  \end{equation}
  \begin{proof}
    Observe that the entries of $W$ are, up to sign, those of $A+B$ as
    defined by Theorem~\ref{thm::equiv}.  Thus $A+B$ has independent
    N(0,1) entries, so we have the same proof as
    Theorem~\ref{thm::rndgap}.
  \end{proof}
\end{thm}

The universal two-dimensional cluster state may be found in polynomial
time using AQC \cite{siu07}.  The one-dimensional cluster state, while
not universal for quantum computing, is useful for gaining intuition
about cluster states \cite{doherty08}.  The third-order interaction
terms in (\ref{eqn::allterms}) are exactly the stabilizers of the
one-dimensional cluster state \cite{pachos04}.  In fact
\begin{equation}
\mathcal{H} = - \sum_{j=2}^{n-1} \sigma^x_{j} \sigma^z_{j+1} \sigma^x_{j+2}
+(-1)^{n-1} \sigma^y_1 \sigma^z_2 \sigma^z_3 \dots \sigma^z_{n-2} \sigma^y_{n-1}
+(-1)^{n-1} \sigma^y_2 \sigma^z_3 \sigma^z_4 \dots \sigma^z_{n-1} \sigma^y_n
\end{equation}
is a Hamiltonian whose ground state is the one-dimensional cluster
state.  By Theorem~\ref{thm::equiv} $A$ and $B$ matrices corresponding
to $\mathcal{H}$ are circulant, and so by Theorem~\ref{thm::circulant}
the one-dimensional cluster state can be realized in polynomial time.

In general, we can define Fermi operators using spin-$S$ operators
provided $2S+1 = 2^n$ for some $n$ \cite{dobrov03}.  Let us consider
$S=3/2$.  Using $n$ spin-3/2 particles, we can define $2n$ Fermi
operators by
\begin{align}
c_{1,j} &= \frac{-1}{\sqrt{3}} S_j^- S_j^z S_j^-
\prod_{k<j} \left[ \frac{5}{4}-\left(S_k^z\right)^2\right] \;,\\
c_{2,j} &= \frac{1}{\sqrt{3}} \left( \frac{1}{2}+S_j^z\right)^2S_j^-
\prod_{k<j} \left[ \frac{5}{4}-\left(S_k^z\right)^2\right] \;,
\end{align}
where $S^x$, $S^y$, and $S^z$ are spin-3/2 operators and $S^{\pm} =
S^x\pm iS^y$.  While the standard Jordan-Wigner transform applied to a
one-dimensional chain of spin-1/2 particles results in a tridiagonal
$B$ matrix, the spin-3/2 transform applied to a one-dimensional
chain of spin-3/2 particles yields a pentadiagonal $B$ matrix.

\section{Conclusion}

We showed that polynomially-large ground-state energy gaps are rare in
many-body quantum Hermitian operators, but the gap is generically
polynomially-large for quadratic fermionic Hamiltonians.  We then
proved that adiabatic quantum computing can realize the ground states
of Hamiltonians with certain random interactions, as well as the
ground states of one, two, and three-dimensional fermionic interaction
lattices, in polynomial time.  Finally, we used the Jordan-Wigner
transformation to show that our results can be restated with Pauli
operators in a surprisingly simple manner, and also consider a related
spin-3/2 transformation.

Since quadratic fermionic Hamiltonian evolutions are classically
simulatable, the adiabatic quantum computations in
Section~\ref{sec::largegaps} are simulatable.  Thus we have provided a
polynomial-time classical algorithm for finding properties of the
ground states of certain random-interaction Hamiltonians and fermionic
interaction lattices in one, two, and three dimensions.

It should be noted that the Jordan-Wigner transformation can be
generalized to higher dimensions, e.g. \cite{wang91}.  Interesting
results may follow from application of these alternate transformations
to our theorems.

Some fermionic systems may only approximately decouple into
``non-interacting quasiparticles'', unlike the exact decouplings
studied here.  These systems may be ``approximately'' classically
simulatable have ``approximately'' polynomially-large ground state
energy gaps.  This may be interesting to explore.

Also, in principle, the decoupling-transformation may be classically
difficult to find, whereas in the Hamiltonians studied here it is
known how to find them.  It is interesting whether similar results
could be obtained for systems where the transformation is difficult to
find explicitly.

\appendix 
\section{Proof of the Lieb {\em et al.} theorem}

For completeness we include a proof of Theorem~\ref{thm::lieb}.  We
first need the property that the fermionic commutation relations are
preserved under certain unitary transformations.

\begin{thm}[Unitary transformations]
  \label{thm::unitaryxform}
  Suppose the operators $\{c_j:\;j=1,\dots, n\}$ obey the
  FCRs.  Let
  \begin{equation}
    T = 
    \left( \begin{array}{cc} 
      U & V \\
      V & U
    \end{array} \right),
  \end{equation}
  where $U$ and $V$ are real $n\times n$ matrices, and suppose $T$ is
  unitary.  Define the set of operators $\{\eta_j:\;j=1,\dots,n\}$ by
  \begin{equation}
    \left(\begin{array}{c} 
      \eta_1\\
      \eta_2\\
      \vdots\\
      \eta_n\\
      \eta_1^{\dagger}\\
      \eta_2^{\dagger}\\
      \vdots\\
      \eta_n^{\dagger}
      \end{array}\right) = 
    T
   \left(\begin{array}{c} 
     c_1\\
     c_2\\
     \vdots\\
     c_n\\
     c_1^{\dagger}\\
     c_2^{\dagger}\\
     \vdots\\
     c_n^{\dagger}
   \end{array}\right) \;,
   \label{eqn::xform}
  \end{equation}
  where (\ref{eqn::xform}) denotes the transformation
    \begin{align}
      \eta_j = \sum_{i=1}^n T_{j,i} c_i + T_{j,i+n} c_{i+n}^{\dagger} \;.
    \end{align}
  Then $\{\eta_j:\;j=1,\dots,n\}$ also obey the FCRs.
  \begin{proof}
    The proof follows from substituting the definitions of
    $\{\eta_j:j=1\dots n\}$ into the FCRs, and using the known
    commutation relations on $\{c_j,c_j^{\dagger}:j=1\dots n\}$.
  \end{proof}
\end{thm}

Now we are ready to prove Theorem~\ref{thm::lieb}:
\begin{proof}[Proof of Theorem~\ref{thm::lieb}]
  We write Equation (\ref{eqn::quadForm}) as
  \begin{align}
    \mathcal{H} &= \left(\begin{array}{cccccccc}
      c_1^{\dagger} & c_2^{\dagger} & ... & c_n^{\dagger} &
      c_1 & c_2 & ... & c_n 
    \end{array}{}\right)
    \left(\begin{array}{rr}
      A &  B \\
      -B & -A
    \end{array}\right)
    \left(\begin{array}{c}
      c_1 \\ c_2 \\ \vdots \\ c_n \\
      c_1^{\dagger} \\ c_2^{\dagger} \\ \vdots \\ c_n^{\dagger}
    \end{array}{}\right) \;.
    \label{eqn::matrixnotation}
  \end{align}
  The theorem is equivalent to showing there are solutions to
  \begin{align}
    &\left(\begin{array}{rr}
      A &  B \\
      -B & -A
    \end{array}\right)
    = \notag \\
    &\frac{1}{2}\left( \begin{array}{rr} 
      (X+Y) & (X-Y) \\
      (X-Y) & (X+Y)
    \end{array} \right)  ^{\dagger} 
    \left( \begin{array}{rr} 
      \Lambda  & 0 \\
      0 & -\Lambda
    \end{array} \right) 
    \frac{1}{2}
    \left( \begin{array}{rr} 
      (X+Y) & (X-Y) \\
      (X-Y) & (X+Y)
    \end{array} \right) \;,
    \label{eqn::equivform}
  \end{align}
  for some non-negative real $n\times n$ diagonal matrix $\Lambda$,
  where $X$ and $Y$ are unitary.  If so, then substituting Equation
  (\ref{eqn::equivform}) into Equation (\ref{eqn::matrixnotation})
  and using the definition of $\eta_k$, we get
  \begin{align}
    \mathcal{H} &= \sum_{k=1}^n \left( \Lambda_k\eta_k^{\dagger}\eta_k  
    - \Lambda_k \eta_k \eta_k^{\dagger} \right) \;.
  \end{align}
  Further, by Theorem~\ref{thm::unitaryxform}, $\{\eta_k:k=1\dots
  n\}$ satisfy the FCRs.  So we can apply the FCRs to the second
  term in each summand to get Equation (\ref{eqn::princip}).

  Now we set about finding solutions to Equation
  (\ref{eqn::equivform}).  We rewrite it for convenience as:
  \begin{equation}
    \left( \begin{array}{rr} 
      X+Y & X-Y \\
      X-Y & X+Y
    \end{array} \right)
    \left(\begin{array}{cc}
      A & B \\
      -B & -A
    \end{array}\right)
    = 
    \left( \begin{array}{cc} 
      \Lambda & 0 \\
      0 & -\Lambda
    \end{array} \right)
    \left( \begin{array}{cc} 
      X+Y & X-Y \\
      X-Y & X+Y
    \end{array} \right) \;. 
    \label{eqn::matrixform}
  \end{equation}
  Equation (\ref{eqn::matrixform}) is equivalent to the
  following four equations:
  \begin{align}
    (X+Y)A  - (X-Y)B    &= \Lambda(X+Y) \label{eqn::first} \;,\\
    (X+Y)B  - (X-Y)A    &= \Lambda(X-Y) \label{eqn::second} \;,\\
    (X-Y)A  - (X+Y)B    &= -\Lambda(X-Y) \label{eqn::third} \;,\\
    (X-Y)B  - (X+Y)A    &= -\Lambda(X+Y) \label{eqn::fourth}\;.
  \end{align}
  Evidently only two of the equations are independent.  Adding and
  subtracting Equations (\ref{eqn::first}) and
  (\ref{eqn::third}) yields
  \begin{align}
    X(A-B) &= \Lambda Y \label{eqn::eq1} \;,\\
    Y(A+B) &= \Lambda X \label{eqn::eq2}\;.
  \end{align}
  We can left-multiply by $\Lambda$ to get
  \begin{align}
    \Lambda X(A-B) &= \Lambda^2 Y \label{eqn::eq3} \;,\\
    \Lambda Y(A+B) &= \Lambda^2 X \label{eqn::eq4}\;, 
  \end{align}
  and then substitute Equation (\ref{eqn::eq2}) into
  Equation (\ref{eqn::eq3}) and Equation (\ref{eqn::eq1}) into
  Equation (\ref{eqn::eq4}) to get the pair of
  eigen-decomposition equations
  \begin{align}
    Y(A+B)(A-B) &= \Lambda^2 Y\label{eqn::eig1} \;,\\
    X(A-B)(A+B) &= \Lambda^2 X \label{eqn::eig2}\;.
  \end{align}
  Since $A$ is real symmetric and $B$ is real anti-symmetric,
  $(A+B)^{\dagger} = A-B$ and so $(A-B)(A+B)$ and $(A+B)(A-B)$ are
  symmetric positive semi-definite.  So there is always a unitary
  $X$ and $Y$ with non-negative diagonal $\Lambda^2$ satisfying
  Equations (\ref{eqn::eig1}) and (\ref{eqn::eig2}).
\end{proof}

\begin{acknowledgments}
  This research is supported in part by National Science Foundation
  Grant CCF 0514213.  The authors would like to thank Stephen Bullock,
  Gavin Brennen, William Kaminsky, Charles Clark, Ana Maria Rey,
  E. ``Manny" Knill, and Eite Tiesinga for
  helpful comments and feedback.
\end{acknowledgments}

\bibliography{sources}
\end{document}